\documentclass[conference]{IEEEtran}
\IEEEoverridecommandlockouts
\usepackage{cite}
\usepackage{amsmath,amssymb,amsfonts}
\usepackage{algorithmic}
\usepackage{algorithm}
\usepackage{graphicx}
\usepackage{textcomp}
\usepackage{xcolor}
\usepackage{booktabs}
\usepackage{placeins}
\usepackage[hidelinks]{hyperref}
\hypersetup{
    pdftitle={A Hybrid Quantum-Classical Coordination Architecture for Portfolio Optimization via Global Context Injection},
    pdfauthor={Xiaoguang Yang and Menghan Dou and Guoping Guo},
    pdfkeywords={Quantum Finance, Portfolio Optimization, Hybrid Quantum-Classical Optimization, Problem Decomposition, QAOA}
}

\newtheorem{proposition}{Proposition}

\def\BibTeX{{\rm B\kern-.05em{\sc i\kern-.025em b}\kern-.08em
    T\kern-.1667em\lower.7ex\hbox{E}\kern-.125emX}}

\begin{document}

\title{A Hybrid Quantum-Classical Coordination Architecture for Portfolio Optimization via Global Context Injection\thanks{*Corresponding authors.}
}

\author{
\IEEEauthorblockN{Xiaoguang Yang\textsuperscript{*}, Menghan Dou\textsuperscript{*}, and Guoping Guo\textsuperscript{*}}
\IEEEauthorblockA{\textit{Origin Quantum Computing Technology (Hefei) Co., Ltd.} \\
Hefei, Anhui, China \\
\{yxg2, dmh, ggp\}@originqc.com}
}

\maketitle

\begin{abstract}
Near-term quantum and quantum-inspired solvers for portfolio optimization rely on local subproblem execution under severe size constraints, but this locality can omit cross-cluster covariance essential for global risk coordination. We propose Context-Aware Folding (CAF), a lightweight hybrid quantum-classical coordination layer between one-shot static decomposition and full-matrix optimization. CAF injects a compressed global risk state into each local subproblem via a state-dependent linear bias, decouples local candidate generation from global commitment, and retains a sequential acceptance rule with a monotonic non-divergence guarantee. On a 2016 Russell 3000 subset (N=484) with Simulated Annealing (SA), CAF improves the scalarized mean-variance objective by 6.59\% over a static baseline (20/20 wins), and by 0.2579\% on an additional 2018 panel (N=1397, 17/20 wins). We further report matched folded N=40 compatibility studies with the Quantum Approximate Optimization Algorithm (QAOA) and simulated quantum annealing (SQA) as local solvers, together with a frozen hardware-in-the-loop run on Origin Quantum's Wukong 180 (\texttt{WK\_C180}) for one warm-started QAOA subproblem at n=5. These results support CAF as a coordination architecture with strong classical large-scale evidence, cross-backend compatibility, and local executability on real quantum hardware in the noisy intermediate-scale quantum (NISQ) era.
\end{abstract}


\section{Introduction}

The Mean-Variance Portfolio Optimization (MVO) problem is a cornerstone of modern quantitative finance \cite{markowitz1952portfolio}. As one of the canonical use cases in quantum finance \cite{orus2019quantum}, MVO becomes NP-hard when cardinality constraints (i.e., limiting the total number of invested assets) are introduced. This makes it a prime candidate for stochastic and quantum optimization heuristics such as Simulated Annealing (SA), Quantum Annealing (QA), and the Quantum Approximate Optimization Algorithm (QAOA) in explicitly financial settings \cite{kirkpatrick1983optimization,farhi2014qaoa,venturelli2019reverse,kerenidis2019quantum}. Recent progress in QAOA has substantially improved the viability of local quantum optimization through warm-start strategies, parameter-transfer or schedule-transfer methods, and large-scale distributed or hardware-aware studies \cite{egger2021warm,cepaite2025quantumenhanced,nzongani2026scaling,xu2025distributed,he2026regularized}. Yet in the dense portfolio settings considered here and under current hardware and simulation constraints, these advances still leave practitioners with restricted local subproblems embedded within a larger hybrid workflow.

However, the limited qubit count and restricted connectivity of current noisy intermediate-scale quantum (NISQ) hardware \cite{preskill2018nisq} pose a severe bottleneck. A prominent industrial approach is to decompose the massive fully-connected asset graph into smaller, disjoint communities after Random Matrix Theory (RMT) preprocessing and graph-based clustering \cite{newman2006modularity,laloux1999noise,jpmorgan2024pipeline,jpmorgan_dcmppln}. While this enables execution on near-term devices, the resulting \textit{separability assumption}---treating each community as independent during optimization---omits cross-cluster covariance terms that remain important for global risk diversification.

To address this, we propose \textbf{CAF (Context-Aware Folding)}. Rather than treating decomposition as a one-shot process, CAF treats the clustered output as an initialization state and focuses on the missing coordination layer after partitioning. It iteratively re-optimizes each community by injecting the current decisions of all other communities as a linear bias (the ``context''), as illustrated in Fig.~\ref{fig:concept}. In this sense, CAF is not a new partitioning scheme but a lightweight coordination layer for reintroducing global risk information into hardware-compatible local solvers.

\begin{figure}[tp]
    \centering
    \includegraphics[width=0.98\columnwidth]{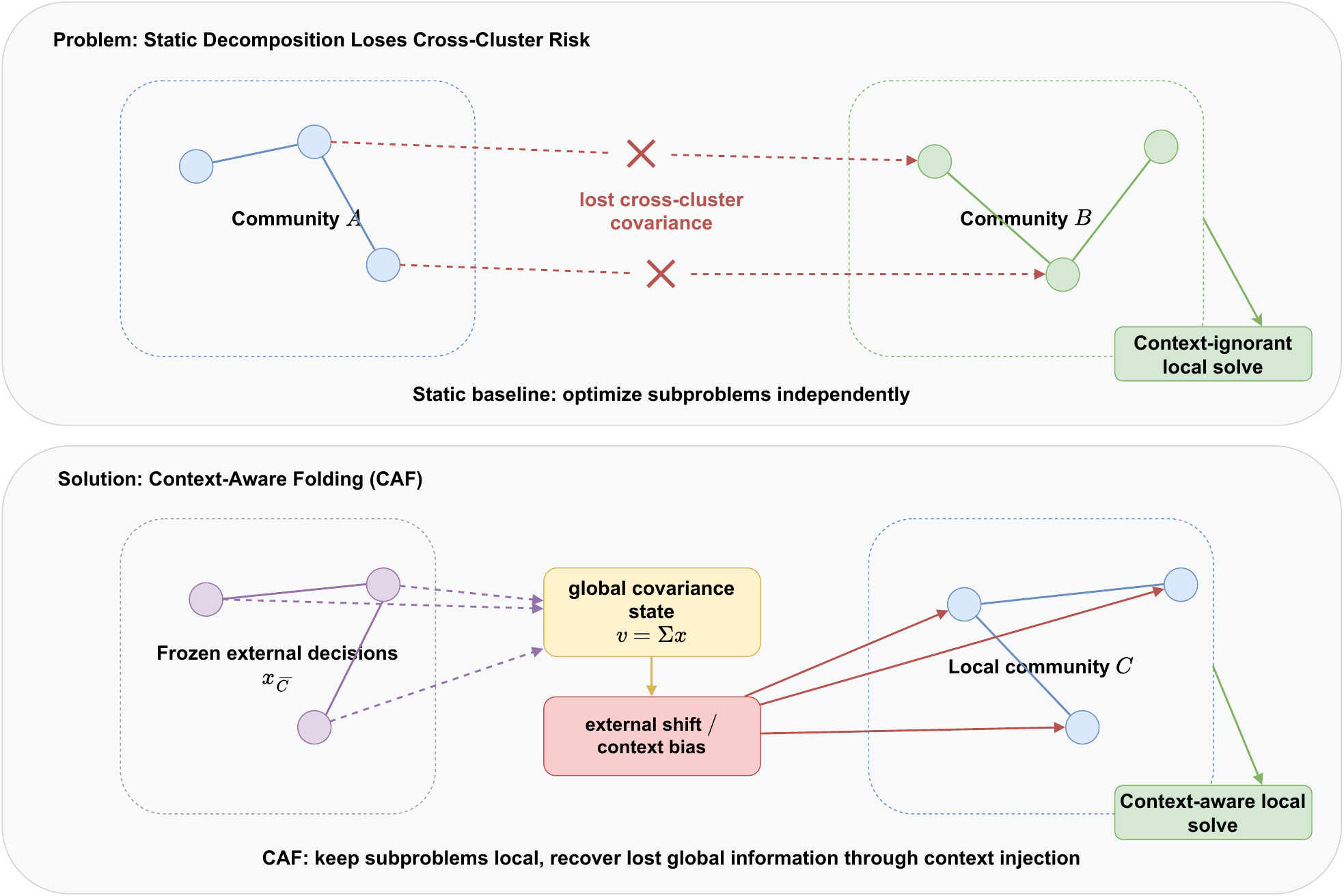}
    \caption{Static decomposition omits cross-cluster covariance terms during local optimization, whereas CAF reintroduces them through a compressed linear context term.}
    \label{fig:concept}
\end{figure}

This paper presents the mathematical formulation of CAF and evaluates it from two complementary perspectives: large-scale architectural validation with classical backends and small-scale compatibility / executability studies with quantum simulators and real quantum hardware. To support reproducibility and facilitate further research, the core CAF architecture engine and the complete benchmark reproduction scripts are open-sourced at \url{https://github.com/kedayxg/CAF}. The main contributions of this work are as follows:
\begin{itemize}
    \item \textbf{Global Context Injection Architecture:} We propose CAF, a hybrid quantum-classical coordination architecture that restores part of the global risk information lost after static partitioning by injecting a state-dependent linear bias into each local portfolio subproblem.
    \item \textbf{Efficient Coordinated Update Rule:} We derive an $O(|C|^2)$ local context extraction procedure and an exact local-form objective increment for sequential commit decisions, yielding a monotonic non-divergence guarantee under the CAF acceptance rule. CAF's decoupled proposal phase is also architecturally compatible with parallel batch execution across independent subproblems.
    \item \textbf{Multi-Scale, Multi-Backend Validation:} On a dense financial benchmark ($N=484$), CAF improves the scalarized objective by $6.59\%$ over the static baseline across 20 SA runs ($20/20$ wins), and the gain persists on a second 2018 panel ($N=1397$) and under a deterministic Gurobi local solver. A frozen-state ablation isolates context injection as the driver, while matched plain-QAOA/SQA studies and a warm-started local update executed on Origin Wukong 180 show that CAF's folded subproblems remain operable on quantum simulators and real quantum hardware without modifying the outer-loop acceptance rule.
\end{itemize}

\subsection{Related Work}

\textbf{Quantum Portfolio Optimization.} Portfolio optimization is a canonical NP-hard problem when cardinality or integer constraints are imposed. Several studies have explored mapping such problems to quantum annealers. For instance, Rosenberg et al. \cite{rosenberg2016solving} and Venturelli et al. \cite{venturelli2019reverse} provided early annealing-based studies of portfolio optimization on D-Wave hardware, including multi-period trading and mean-variance settings. However, these formulations often induce dense logical interaction graphs, which in turn lead to substantial minor-embedding overhead on sparse physical hardware topologies \cite{choi2008embedding} (where multiple physical qubits must be chained to represent a single logical qubit). On the gate-based side, theoretical quantum algorithms have been proposed for constrained portfolio optimization \cite{kerenidis2019quantum}. Beyond the original QAOA proposal \cite{farhi2014qaoa}, constrained extensions such as the Quantum Alternating Operator Ansatz \cite{hadfield2019qaoa} and warm-start strategies \cite{egger2021warm} have been developed for combinatorial optimization, and more recent work has pushed QAOA further through warm-started hybrid workflows, schedule transfer, and distributed large-scale execution \cite{cepaite2025quantumenhanced,nzongani2026scaling,xu2025distributed,he2026regularized}. While empirical portfolio-specific QAOA studies have appeared \cite{brandhofer2022benchmarking}, more recent work has extended quantum portfolio formulations to higher-order moments \cite{uotila2025higher} and to constraint-aware encodings that avoid penalty-induced landscape distortion \cite{thomassin2025constrained,hao2026constraint}; nonetheless, these studies remain focused on either relatively small financial instances, graph-structured benchmarks, or simulation-backed scaling, and an extensive recent benchmark reports only very limited room for a potential quantum advantage over state-of-the-art classical solvers on this problem \cite{stopfer2025benchmark}. This is exactly where our motivation differs: even when local QAOA becomes more usable, dense Markowitz quadratic unconstrained binary optimization (QUBO) problems remain globally coupled through covariance interactions and cardinality consistency, so locally strong candidates still need a coordination layer before they can be committed safely at the portfolio level. Ultimately, current gate errors and limited coherence times in the NISQ regime \cite{preskill2018nisq} still restrict the practical execution of these local gate-based solvers to very small asset universes.

\textbf{Graph Decomposition in NISQ.} To circumvent hardware limits, hybrid algorithms heavily rely on problem decomposition. While classical optimization often employs continuous relaxation techniques such as Benders or Dantzig-Wolfe decomposition \cite{benders1962partitioning,dantzig1960decomposition}, these methods map poorly to discrete quantum hardware. Consequently, in the quantum domain, graph-based clustering has become a common approach \cite{booth2017partitioning}. A recent industrial example is JPMorgan's decomposition pipeline \cite{jpmorgan2024pipeline,jpmorgan_dcmppln}, which applies RMT preprocessing and modified spectral clustering based on Newman's method to partition the asset correlation graph. Each cluster is then solved independently on a quantum or classical heuristic backend; very recent work has executed such RMT-and-community-detection decomposition pipelines end-to-end on trapped-ion quantum processors, where each cluster defines a hardware-embeddable QUBO subproblem \cite{gomez2026largescale}. Although some pipelines incorporate static risk-rebalancing or lightweight heuristic post-processing to mitigate boundary effects \cite{jpmorgan_dcmppln}, the hard partitioning still omits cross-cluster covariance terms during the core optimization phase, so the local solvers no longer account for the full global risk interaction structure.

\textbf{Iterative Refinement and Coordinated Decomposition.} Classical coordination frameworks, such as Block Coordinate Descent (BCD) \cite{tseng2001bcd} and Lagrangian-decomposition methods \cite{swoboda2017dual,martins2015ad3}, provide mechanisms for coordinating coupled subproblems. However, adapting these classical dual-ascent methods directly to our QUBO pipeline is problematic, as they typically require continuous variable relaxations or generic consensus variables that are difficult to map directly onto current hardware-compatible binary formulations. While Qbsolv-style hybrid sub-QUBO heuristics \cite{booth2017partitioning} are well known, they operate on raw QUBO partitions and focus on heuristic extraction/merge strategies rather than domain-specific global coordination. To bridge this gap, our Context-Aware Folding (CAF) operates explicitly at the financial domain level. Instead of passing generic continuous dual variables, it computes a state-dependent ``risk shift'' vector from the current global portfolio via a local $O(|C|^2)$ computation and injects it directly into the local linear term of the QUBO. This keeps CAF compatible with existing local quantum solvers at the formulation level while preserving a direct interpretation in terms of global portfolio risk.

\section{The CAF Architecture}

The design of the CAF architecture is inspired by classical BCD and Large Neighborhood Search (LNS) \cite{shaw1998lns}. A naive adaptation of these iterative methods to portfolio optimization would require repeatedly evaluating the full $O(N^2)$ global objective to assess the impact of every single asset flip. Furthermore, standard BCD is inherently sequential; updating one variable block strictly depends on the latest state of the others. In the context of quantum computing, a purely sequential outer loop can be highly inefficient for hardware utilization, as it prevents the batch submission of multiple independent subproblems to the Quantum Processing Unit (QPU) or distributed classical solvers, forcing the system to incur heavy latency and queueing overheads. 

To address these critical bottlenecks, CAF improves upon the classical coordinate descent paradigm. By explicitly decoupling the proposal phase from the commit phase, CAF exposes a parallelizable proposal stage within each iteration where multiple disjoint communities can generate candidate updates simultaneously. This architecture is designed to support batch submission to independent backends; in an actual concurrent deployment, the solver-side proposal latency would be governed by the slowest subproblem in each iteration rather than by the sum over all subproblems. The present experiments do not benchmark such concurrent execution directly and instead focus on objective quality and architectural behavior. Meanwhile, the incremental cache ensures that the classical coordinator does not become a bottleneck by replacing repeated global $O(N^2)$ evaluations with local $O(|C|^2)$ shift extraction. The complete architecture and its implementation details are described below.

\begin{figure*}[tp]
    \centering
    \includegraphics[width=0.98\textwidth]{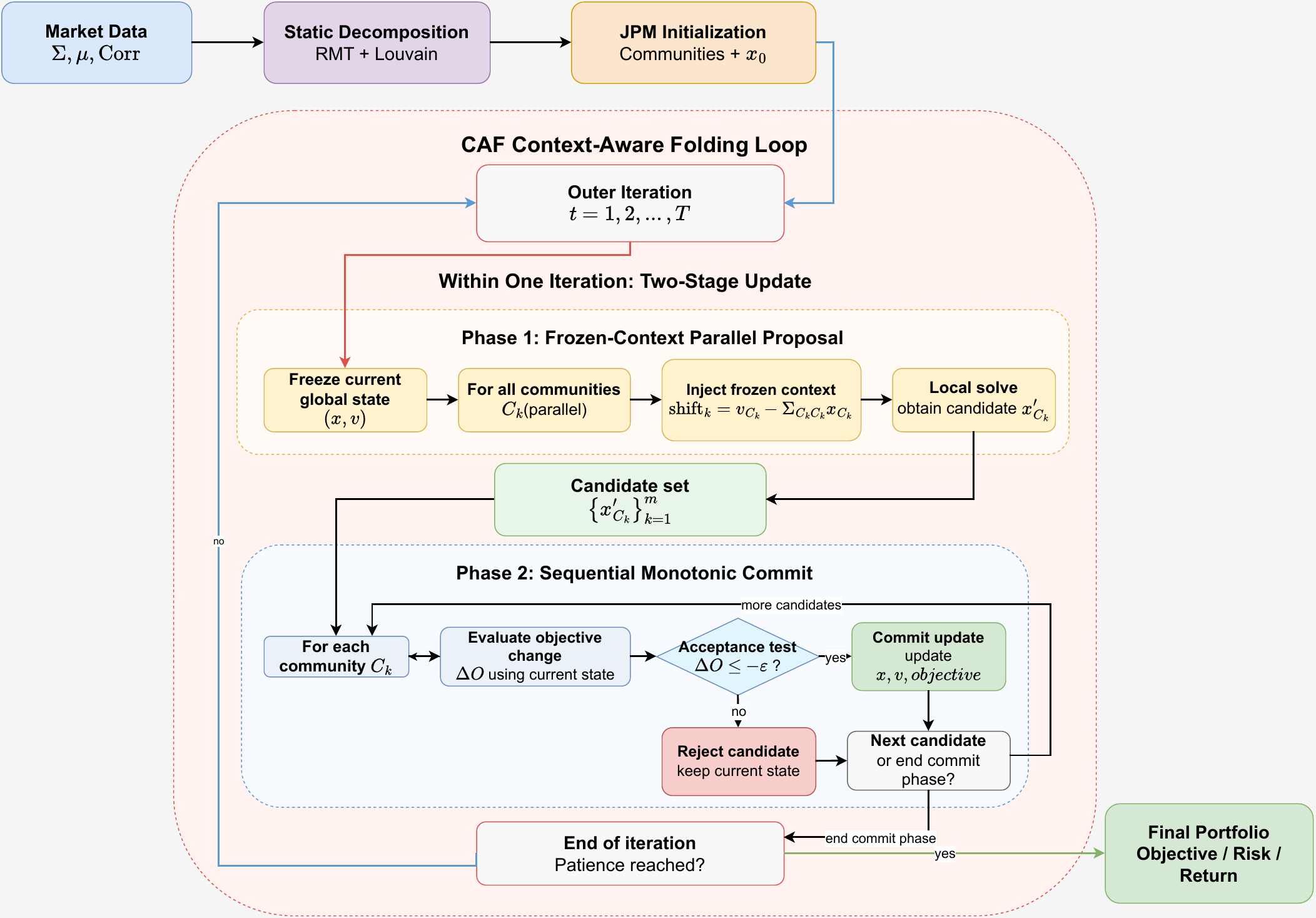}
    \caption{CAF workflow: static decomposition initializes the portfolio, parallel proposals are generated under a frozen context, and updates are committed sequentially by a monotonic test.}
    \label{fig:overall_flow}
\end{figure*}

This diagram (Fig.~\ref{fig:overall_flow}) summarizes the full system pipeline. Raw market inputs are first processed by an upstream static decomposition stage to produce communities and an initial portfolio, after which CAF enters an outer iteration loop. Within each iteration, CAF first generates community-wise proposals under a frozen global context, then applies a sequential monotonic acceptance step against the current global state, and only performs the patience-based stopping test after the commit phase is completed.

\subsection{The Global Objective}
The global portfolio objective under cardinality constraint $K$ is defined as:
\begin{equation}
\label{eq:global_objective}
\min_{x \in \{0,1\}^N} \mathcal{O}(x) = q \cdot x^T \Sigma x - \mu^T x \quad \text{s.t.} \sum_i x_i = K
\end{equation}
where $\Sigma$ is the covariance matrix, $\mu$ is the expected return vector, and $q$ is the risk aversion factor.
When this constrained objective is passed to a hardware-compatible local solver, the community-level budget constraint is absorbed into the local QUBO through a standard quadratic penalty term $\lambda \left(\sum_{i \in C} x_i - K_C\right)^2$ with sufficiently large $\lambda$. CAF therefore modifies only the effective linear bias term seen by each local solver, while leaving the underlying constrained subproblem structure unchanged.

\subsection{Context Injection (The External Shift)}
Given a community $C$ and the set of all other assets $\bar{C}$, the global risk term $x^T \Sigma x$ can be decomposed. When optimizing $x_C$ while holding $x_{\bar{C}}$ fixed, the cross-term $2 x_C^T \Sigma_{C\bar{C}} x_{\bar{C}}$ acts as a linear penalty imposed by the outside world.

To avoid repeatedly computing this cross-term via explicit dependence on the full complement set $\bar{C}$ during each iteration, CAF maintains a global covariance state vector, defined as:
\begin{equation}
v = \Sigma x \in \mathbb{R}^N
\end{equation}
Here, the $i$-th element of $v$ represents the total covariance contribution from the currently selected global portfolio to asset $i$. When focusing on a specific community $C$, we extract the corresponding sub-vector $v_C \in \mathbb{R}^{|C|}$. Mathematically, $v_C$ aggregates the covariance from all assets, such that $v_C = \Sigma_{CC} x_C + \Sigma_{C\bar{C}} x_{\bar{C}}$. 

By exploiting this property, we can isolate the external interference simply by subtracting the internal community covariance.

\begin{proposition}[$O(|C|^2)$ Local Shift Extraction]
Given the global covariance state vector $v_C$, the external shift for community $C$ can be computed in $O(|C|^2)$ time as:
\begin{equation}
\text{Shift}_C = v_C - \Sigma_{CC} x_C = \Sigma_{C\bar{C}} x_{\bar{C}}
\end{equation}
\end{proposition}
\begin{IEEEproof}[Proof Sketch]
By definition, $v_C$ is the $C$-indexed sub-vector of $\Sigma x$. Expanding the matrix-vector product yields $v_C = \Sigma_{CC} x_C + \Sigma_{C\bar{C}} x_{\bar{C}}$. Subtracting the internal covariance term $\Sigma_{CC} x_C$, which takes $O(|C|^2)$ operations, directly isolates the external interference $\Sigma_{C\bar{C}} x_{\bar{C}}$ without needing to iterate over the potentially massive set $\bar{C}$.
\end{IEEEproof}

Intuitively, $\text{Shift}_C$ is the linear risk penalty induced by assets outside community $C$ under the current global portfolio state. Because $x_{\bar{C}}$ is frozen during the local optimization of community $C$, the purely external risk term $q x_{\bar{C}}^T \Sigma_{\bar{C}\bar{C}} x_{\bar{C}}$ becomes a constant and can be ignored. Consequently, the global cross-term $2q x_C^T \Sigma_{C\bar{C}} x_{\bar{C}}$ can be absorbed into the linear portion of the local subproblem. Using the extracted context, the adjusted linear returns for the local sub-QUBO become:
\begin{equation}
\label{eq:adjusted_returns}
\tilde{\mu}_C = \mu_C - 2q \cdot \text{Shift}_C
\end{equation}

\subsection{Monotonic Acceptance and Dynamic Annealing}
Once a local quantum or simulated annealer proposes a new state $x'_C$ satisfying the local budget constraint, CAF must evaluate the exact global objective change $\Delta \mathcal{O} = \mathcal{O}(x') - \mathcal{O}(x)$ to decide whether to accept the update. Let the local state change be $\Delta x_C = x'_C - x_C$. Because all external variables $x_{\bar{C}}$ remain unchanged during this local update, the global quadratic expansion simplifies. By substituting $x' = x + \Delta x$ into the global objective and factoring out the cross-terms using the previously defined state vector $v_C = \Sigma_{CC} x_C + \Sigma_{C\bar{C}} x_{\bar{C}}$, the exact global increment reduces to an efficient local computation:
\begin{equation}
\Delta \mathcal{O} = q \left(\Delta x_C^T \Sigma_{CC} \Delta x_C + 2 v_C^T \Delta x_C\right) - \mu_C^T \Delta x_C
\end{equation}
The update is accepted if and only if $\Delta \mathcal{O} \le -\epsilon$, where $\epsilon=10^{-5}$ is a numerical tolerance. 

\begin{proposition}[Monotonic Non-Divergence]
Under the acceptance criterion $\Delta \mathcal{O} \le -\epsilon$, every committed CAF update decreases the global objective by at least the prescribed numerical tolerance. Since the search space $\{0,1\}^N$ is finite, CAF cannot diverge to arbitrarily worse states; instead, it generates a finite descending sequence of accepted objective values until no further tolerance-significant improvement is accepted.
\end{proposition}
\begin{IEEEproof}[Proof Sketch]
The state $x \in \{0,1\}^N$ resides in a finite discrete space, and $\mathcal{O}(x)$ is deterministically evaluated. A candidate is generated under a frozen global context and committed only if $\Delta \mathcal{O} = \mathcal{O}(x') - \mathcal{O}(x) \le -\epsilon$, so every committed update lowers the evaluated global objective by at least $\epsilon$ and the accepted objective sequence is strictly descending up to the tolerance. A finite space cannot admit an infinite sequence of distinct $\epsilon$-decreases, so CAF reaches a state from which no further tolerance-significant update is committed. As the rule is greedy, CAF is a coordinated local-refinement procedure rather than a globally convergent exact method.
\end{IEEEproof}
By construction, CAF is a monotone block-restricted local-refinement procedure: within each outer iteration, it forms community-wise proposals under a frozen context and commits a proposal only if it strictly decreases the current global objective. Under exact local subproblem solves, this yields a block-coordinate descent scheme whose terminal state is locally optimal with respect to single-community exact updates; with heuristic local solvers such as SA, QAOA, or simulated quantum annealing (SQA), the procedure instead terminates at a block-restricted fixed point under the chosen local solver. The magnitude of CAF's improvement is therefore governed by how much cross-cluster structure the upstream static decomposition leaves uncoordinated (the available headroom), rather than by any guarantee of global optimality.

For the annealing-based local solves in the main SA benchmark, CAF also uses \textit{Dynamic Reads}, a simple three-level size-adaptive schedule that assigns 20, 100, or 200 reads according to community size, with a per-local-subproblem-solve cap matching the 200-read setting used in the main SA baseline. The schedule does not increase the maximum local solve budget beyond the baseline setting; it only allocates fewer reads to smaller communities whose local search spaces are correspondingly easier. We do not apply this schedule to gate-based QAOA, because QAOA is not parameterized by annealing reads: its computational budget is governed instead by circuit depth, shot count, and the classical parameter-optimization loop.

\subsection{Algorithm Implementation}
Algorithm \ref{alg:caf} lists the two-stage CAF loop: parallel proposal generation under a frozen context, followed by sequential commit decisions under the current global state.

\begin{algorithm}[tp]
\caption{Context-Aware Folding (CAF)}
\label{alg:caf}
\begin{algorithmic}[1]
\REQUIRE $\Sigma,\ \mu,\ q,\ \{C_k\},\ \{K_k\},\ x^{(0)},\ T,\ P,\ \epsilon$
\STATE $x \leftarrow x^{(0)}$
\STATE $v \leftarrow \Sigma x$
\STATE $\text{best} \leftarrow \mathcal{O}(x)$
\STATE $n_{\mathrm{stall}} \leftarrow 0$
\FOR{$t = 1$ \TO $T$}
    \STATE Freeze current global state $(x,v)$
    \FOR{\textbf{each community} $C_k$ \textbf{in parallel}}
        \STATE Inject frozen context: $\text{Shift}_{C_k} \leftarrow v_{C_k} - \Sigma_{C_k C_k}x_{C_k}$
        \STATE Solve local subproblem on $C_k$ under budget $K_k$ to obtain candidate $x'_{C_k}$
    \ENDFOR
    \FOR{\textbf{each community} $C_k$}
        \STATE $\Delta x_{C_k} \leftarrow x'_{C_k} - x_{C_k}$
        \STATE $\Delta \mathcal{O} \leftarrow q\left(\Delta x_{C_k}^{T}\Sigma_{C_k C_k}\Delta x_{C_k} + 2v_{C_k}^{T}\Delta x_{C_k}\right) - \mu_{C_k}^{T}\Delta x_{C_k}$
        \IF{$\Delta \mathcal{O} \le -\epsilon$}
            \STATE $x_{C_k} \leftarrow x'_{C_k}$
            \STATE $v \leftarrow v + \Sigma_{:,C_k}\Delta x_{C_k}$
        \ENDIF
    \ENDFOR
    \IF{$\mathcal{O}(x) < \text{best}$}
        \STATE $\text{best} \leftarrow \mathcal{O}(x)$
        \STATE $n_{\mathrm{stall}} \leftarrow 0$
    \ELSE
        \STATE $n_{\mathrm{stall}} \leftarrow n_{\mathrm{stall}} + 1$
    \ENDIF
    \IF{$n_{\mathrm{stall}} \ge P$}
        \STATE \textbf{stop}
    \ENDIF
\ENDFOR
\RETURN $x$
\end{algorithmic}
\end{algorithm}

This algorithm matches the implementation used in our experiments. The proposal phase is naturally parallelizable, but the commit phase must remain sequential: if interacting communities were updated simultaneously, their combined cross-cluster penalties could increase the global objective. Re-evaluating $\Delta \mathcal{O}$ against the updated global state $v$ after each accepted move preserves monotonicity. From the coordinator perspective, shift extraction and acceptance testing for one community cost $O(|C_k|^2)$, so one full commit pass costs $O\!\left(\sum_k |C_k|^2\right)$ rather than repeatedly recomputing the global quadratic objective at $O(N^2)$ per local update. If proposal generation is run concurrently, its wall-clock latency is governed by the slowest subproblem. The procedure stops when convergence or the patience criterion is reached.

\section{Experimental Results}

We evaluate CAF at two complementary scopes: a large-scale classical benchmark of its coordination gains, and a small-scale quantum compatibility study on folded subproblems. For the large-scale benchmark, we use a dense 2016 Russell 3000 subset ($N=484$ assets), plus an additional 2018 panel ($N=1397$) to test robustness. In the main benchmarks, the portfolio cardinality constraint is set to $K=N/2$, matching the upstream pipeline's default cardinality setting (e.g., $K=242$ for the 2016 instance, $K=698$ for the 2018 panel), representing a highly constrained and NP-hard optimization scenario.

In the following subsections, we first detail the experimental setup and software backends, then present the core benchmark performance and evaluate the architecture's scalability alongside an architectural ablation. We next assess whether the observed gains persist under deterministic local solves via exact classical cross-validation, before analyzing how CAF reshapes the local optimization landscape and examining folded subproblems on quantum simulators and real quantum hardware, including both matched-backend compatibility results and frozen local hardware validation.

\subsection{Experimental Setup}
We conduct the main large-scale experiments on the preprocessed Russell 3000 benchmark instance released in JPMorgan's open-source decomposition repository \cite{jpmorgan_dcmppln}, which accompanies its decomposition pipeline studies \cite{jpmorgan2024pipeline}. In our reproduction, we use the released instance corresponding to a cleaned cross-sectional snapshot anchored at the beginning of January 2016. The upstream release already filters out assets with insufficient liquidity or missing price data, yielding a dense universe of $N=484$ assets. We take the released annualized expected return vector as $\mu$ and the released covariance matrix as $\Sigma$. Before graph construction, we apply RMT filtering to the empirical correlation matrix, and we then apply the upstream spectral community-detection pipeline on the filtered graph to partition the universe into 9 subproblems. To test robustness beyond a single historical panel, we additionally generated a second benchmark from a separate 2018 Russell 3000 return panel using the same covariance / correlation / returns construction and the same CAF-vs-static evaluation pipeline; after preprocessing, this second panel contains $N=1397$ assets. For reproducibility, the seeded numerical results reported here---both the classical benchmarks and the quantum-simulator compatibility studies---can be reproduced under the pinned environment shipped with the code release. These software-based experiments do not require specialized hardware and can be run on a standard laptop; the reported numbers are controlled by the pinned environment and seeds, while wall-clock time may vary across machines.

For the main benchmark, Simulated Annealing (SA) serves as the primary local backend because it lets CAF and the static baseline share the same scalable solver family at full scale. We then use Gurobi as a deterministic exact cross-check. The quantum study later in this section is more limited in scope and consists of matched-backend compatibility experiments on folded $N=40$ instances plus a separate frozen-state hardware executability study on selected warm-started subproblems. In the hardware part, we freeze a CAF state, build the injected local subproblem, optimize warm-started QAOA parameters once on an ideal simulator, and then execute the same ansatz on Origin Wukong 180 (backend \texttt{WK\_C180}); hardware candidates are decoded under the same feasible-sample rule as in the corresponding simulator protocol. Because hardware calibrations and noise fluctuate over time, the real-hardware raw counts are protocol-reproducible but are not expected to be bitwise identical across reruns. Exact cross-validation used Gurobi Optimizer (v10.0), the gate-based study used custom Qiskit- and PyQPanda-based QAOA workflows, and the annealing-style simulated study used D-Wave's Ocean SDK (\texttt{dwave-samplers}).

The primary baseline is the official static decomposition pipeline from JPMorgan's repository \cite{jpmorgan_dcmppln}, denoted herein as the static baseline. We omit a Full-Matrix Simulated Annealing comparison from the main results, as native execution on the fully-connected $N=484$ graph degrades in solution quality and fails to scale efficiently, which is precisely the motivation for decomposition. For a rigorous evaluation, both the static baseline and CAF utilized exactly the same underlying community structures and the same local solver family (Simulated Annealing). In the main benchmark, the reproduced static-baseline local path follows the upstream static-pipeline heuristic, including its local risk-rescaling step, and the frozen-snapshot quantum studies retain the same local convention because they build directly on that static initialization. By contrast, the scaling, sparsity, and suppression studies keep the CAF local updates and the global acceptance/evaluation aligned under a shared risk-aversion setting $q=0.5$; except in the dedicated $q$-sweep analysis, this same $q=0.5$ setting is also used for the reported CAF local updates and global evaluations discussed below. We therefore describe CAF using the $q=0.5$ notation of Eqs.~\eqref{eq:global_objective}--\eqref{eq:adjusted_returns}. In the main benchmark, the static baseline utilized a fixed 200-read setting for each local subproblem solve. In contrast, CAF employed a piecewise size-adaptive schedule that assigned 20, 100, or 200 annealing reads according to community size (specifically, 20 reads for $|C| \le 30$, 100 reads for $30 < |C| \le 80$, and 200 reads for $|C| > 80$), so that no CAF local subproblem solve exceeded the baseline's per-local-subproblem-solve read cap. This annealing-read schedule should therefore be understood as a size-adaptive redistribution of local sampling effort rather than as a larger per-subproblem budget.

\subsection{Benchmark Performance}
Table \ref{tab:benchmark} presents the core 2016 benchmark results evaluated under a risk aversion factor of $q=0.5$. The metrics reflect the mean and standard deviation over 20 independent stochastic runs.

\begin{table}[tp]
\caption{Benchmark Results on the 2016 Panel ($q=0.5$, Averaged over 20 runs)}
\label{tab:benchmark}
\begin{center}
\resizebox{\columnwidth}{!}{%
\begin{tabular}{lcccc}
\toprule
\textbf{Method} & \textbf{Objective ($\downarrow$)} & \textbf{Risk} & \textbf{Exp. Return} & \textbf{Time (s)} \\
\midrule
Static Baseline & $1.5453 \pm 0.0212$ & $3.3666 \pm 0.0475$ & $0.1380 \pm 0.0031$ & $11.39 \pm 1.39$ \\
\textbf{CAF (Ours)} & $\mathbf{1.4433 \pm 0.0148}$ & $\mathbf{3.1214 \pm 0.0328}$ & $0.1174 \pm 0.0052$ & $22.70 \pm 3.15$ \\
\bottomrule
\end{tabular}%
}
\par\smallskip\footnotesize \textit{Note:} Wall-clock times are indicative only and may vary across machines and runtime environments.
\end{center}
\end{table}

Empirically, CAF consistently achieves a lower objective value compared to the static baseline on this 2016 benchmark (an average relative improvement of $6.59\%$). Across 20 paired runs, CAF improves objective in $20/20$ cases, with an exact two-sided sign-test p-value of $1.91\times10^{-6}$ and a bootstrap $95\%$ confidence interval of $[5.99\%,\,7.22\%]$ for the mean relative improvement. This stable improvement is consistent with the view that CAF partially mitigates the cross-cluster information loss introduced by static decomposition, driving the system to a more favorable configuration. Fig.~\ref{fig:obj_conv} tracks the mean objective convergence across runs, with a shaded $\pm 1$ standard deviation band.

To test whether the gain survives beyond one historical panel, we repeated the same CAF-vs-static evaluation on the additional 2018 full-panel benchmark ($N=1397$, 20 runs, same $q=0.5$ setting). On this second panel, the static baseline attains an objective of $55.1324 \pm 0.2043$, while CAF reaches $54.9900 \pm 0.1792$, corresponding to a mean relative improvement of $0.2579\%$. CAF improves 17 of the 20 runs, with a two-sided sign-test p-value of $0.0026$ and a bootstrap $95\%$ confidence interval of $[0.1756\%,\,0.3458\%]$. The smaller gain reflects a smaller amount of available headroom rather than a failure of the mechanism: on this panel the static decomposition is already closer to a block-restricted fixed point under the current local solver, so CAF accepts fewer outer-loop updates and a minority of runs accept no move at all. The robustness takeaway is therefore that CAF's benefit is not tied to a single market snapshot, even though its magnitude varies with the regime.

\begin{figure}[tp]
    \centering
    \includegraphics[width=0.98\linewidth]{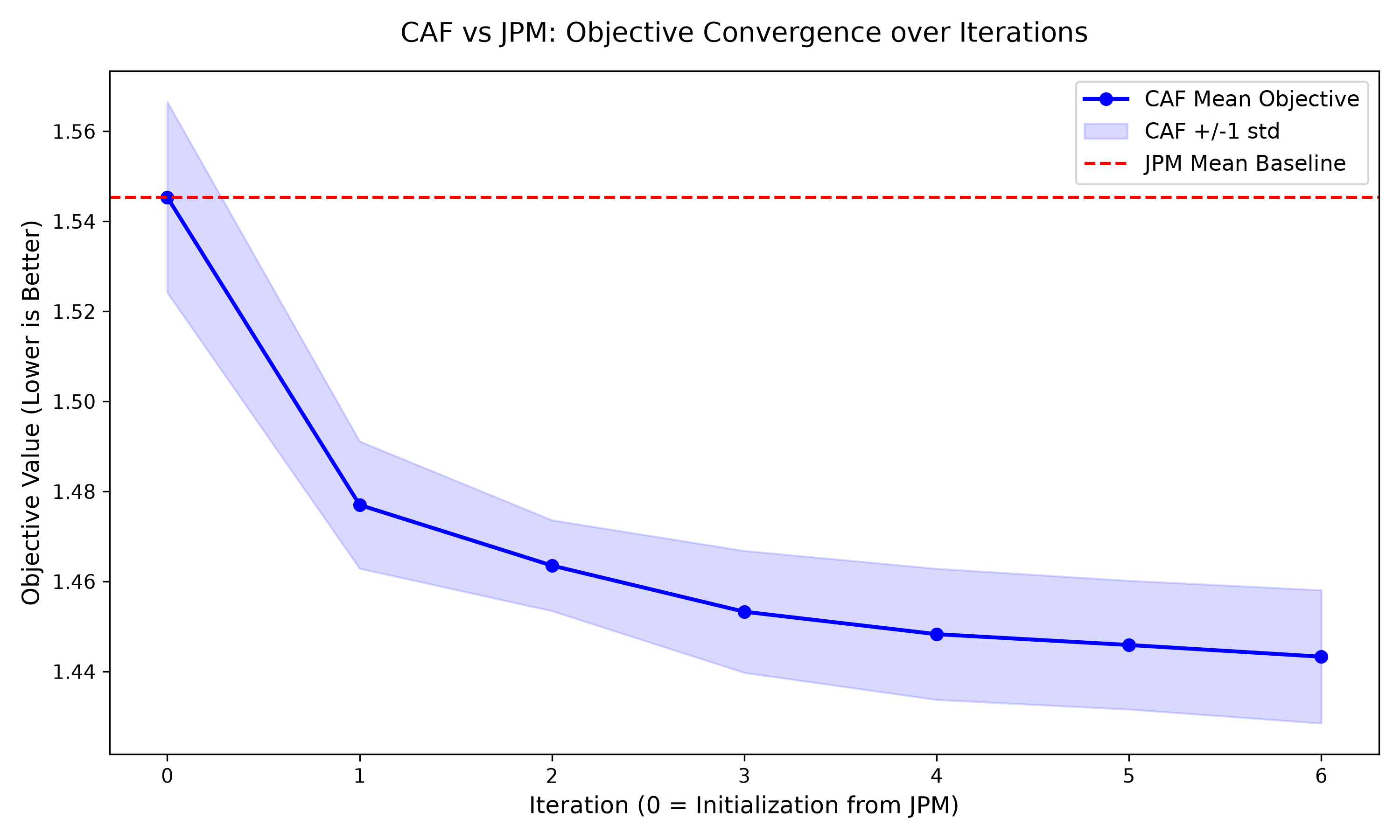}
    \caption{Objective convergence on the 2016 $N=484$ SA benchmark at $q=0.5$ over 20 runs. The shaded region shows $\pm 1$ standard deviation.}
    \label{fig:obj_conv}
\end{figure}

\begin{figure}[tp]
    \centering
    \includegraphics[width=0.98\linewidth]{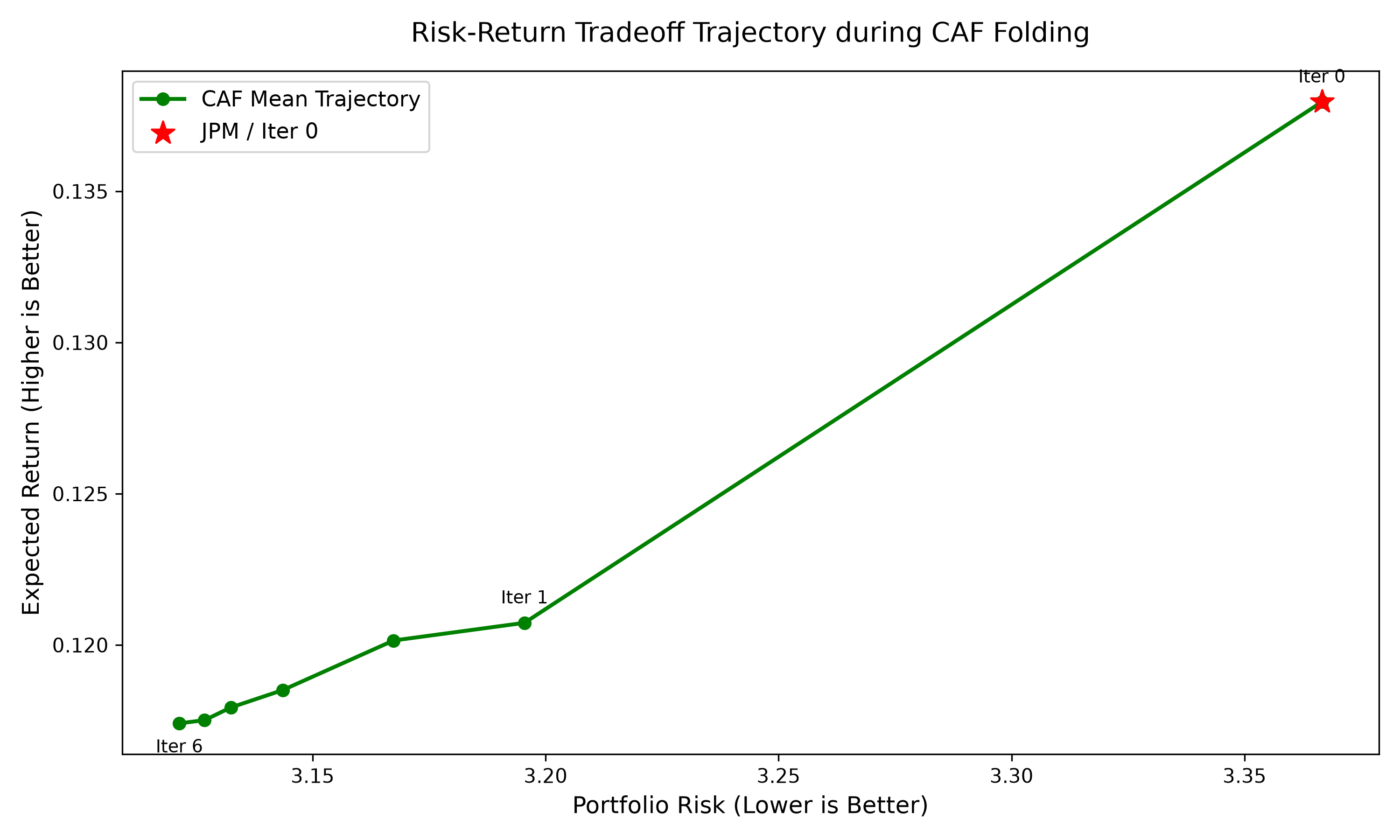}
    \caption{Mean risk-return trajectory on the 2016 $N=484$ SA benchmark at $q=0.5$. Iteration 0 is the static baseline, and later CAF iterations move toward lower risk with a moderate return reduction.}
    \label{fig:risk_ret}
\end{figure}

As shown in Fig.~\ref{fig:risk_ret}, the mean optimization trajectory indicates that CAF systematically moves the portfolio toward a lower-risk region of the risk-return plane after initialization. Although this transition is accompanied by a reduction in expected return (from $0.1380$ to $0.1174$), it is important to emphasize that the underlying solver is optimizing the scalarized objective $\mathcal{O}(x) = q \cdot \text{Risk} - \text{Return}$. At $q=0.5$, the substantial decrease in risk (from $3.3666$ to $3.1214$) mathematically outweighs the drop in return, driving the net objective lower. This behavior is consistent with the possibility that static decomposition underestimates some cross-cluster risk interactions; by omitting these penalties during the local solve stage, it may favor portfolios with higher individual returns but larger unmeasured global risk contributions. The static baseline point is reported explicitly as the shared initialization state at Iteration 0, clarifying that the subsequent trajectory reflects iterative refinement rather than an independent optimization run. We therefore interpret the reported percentage gains strictly under the optimized mean-variance objective and not as direct claims about Sharpe ratio or deployable trading performance.

\subsection{Scalability Across Problem Size}
To complement the dedicated $N=484$ benchmark above, we additionally ran a size sweep on subinstances sampled from the same 2016 benchmark with $N\in\{100,200,300,484\}$ and $K=N/2$, keeping the same SA backend and CAF hyperparameters. For each size, we performed 20 stochastic runs and compared CAF against the static decomposition baseline under identical settings.

These SA scaling results show that CAF retains a positive objective advantage across all tested scales, with mean relative improvements of $22.84\%\pm3.44\%$ ($N=100$), $10.52\%\pm3.05\%$ ($N=200$), $10.52\%\pm2.75\%$ ($N=300$), and $5.80\%\pm1.54\%$ on the independent $N=484$ scaling run. The canonical $N=484$ headline result in this paper remains the dedicated benchmark in Table~\ref{tab:benchmark}, which reports a $6.59\%$ average improvement over 20 runs; the separate $N=484$ scaling run is included only as a consistency check under independent stochastic draws. Runtime increases with problem size for both methods. At small scale, CAF is comparable to or slightly faster than the static baseline ($0.61$s vs. $0.71$s at $N=100$, reflecting the lower adaptive read budget used on the small communities), whereas at the largest tested size it incurs a larger coordination overhead ($34.70$s vs. $8.15$s at $N=484$). These results indicate that CAF's gain is not restricted to a single fixed dimension and remains positive across the tested scales as the decomposed instances grow larger.

To isolate the contributions of specific architectural components, we conducted an ablation study comparing the static baseline, CAF without the incremental context cache (evaluating the full $O(N^2)$ objective repeatedly), and the standard CAF variant with incremental update guarantees.

\begin{figure}[tp]
    \centering
    \includegraphics[width=0.98\linewidth]{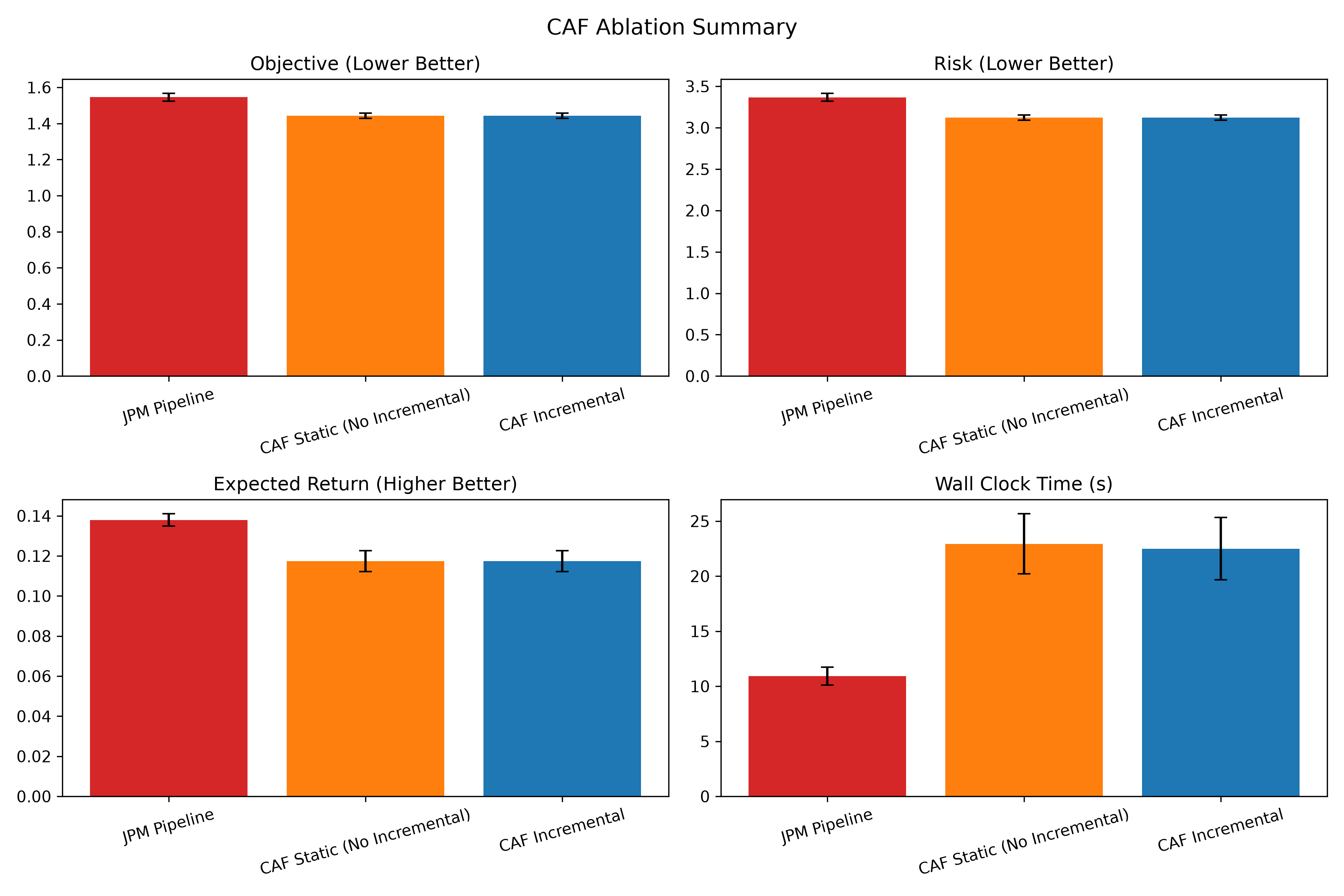}
    \caption{Ablation summary for the static baseline, CAF without cache, and standard CAF. Error bars denote $\pm 1$ standard deviation.}
    \label{fig:ablation}
\end{figure}

Fig.~\ref{fig:ablation} supports the interpretation that the context-aware iterative mechanism itself drives the objective gain, as the cached and non-cached CAF variants reach identical final portfolios. Replacing repeated full-state re-evaluation with an incremental $v$-vector update reduces the acceptance-step cost from a global $O(N^2)$ computation to a local $O(|C|^2)$ extraction, but yields only a modest runtime difference in the present experiments. This indicates that the cache mainly reduces classical coordinator overhead. The parallelizable proposal stage discussed earlier in the architecture description should therefore be interpreted as a deployment-level scalability opportunity under concurrent execution, rather than as an empirically benchmarked source of speedup in the present sequential implementation.

\begin{figure}[tp]
    \centering
    \includegraphics[width=0.98\linewidth]{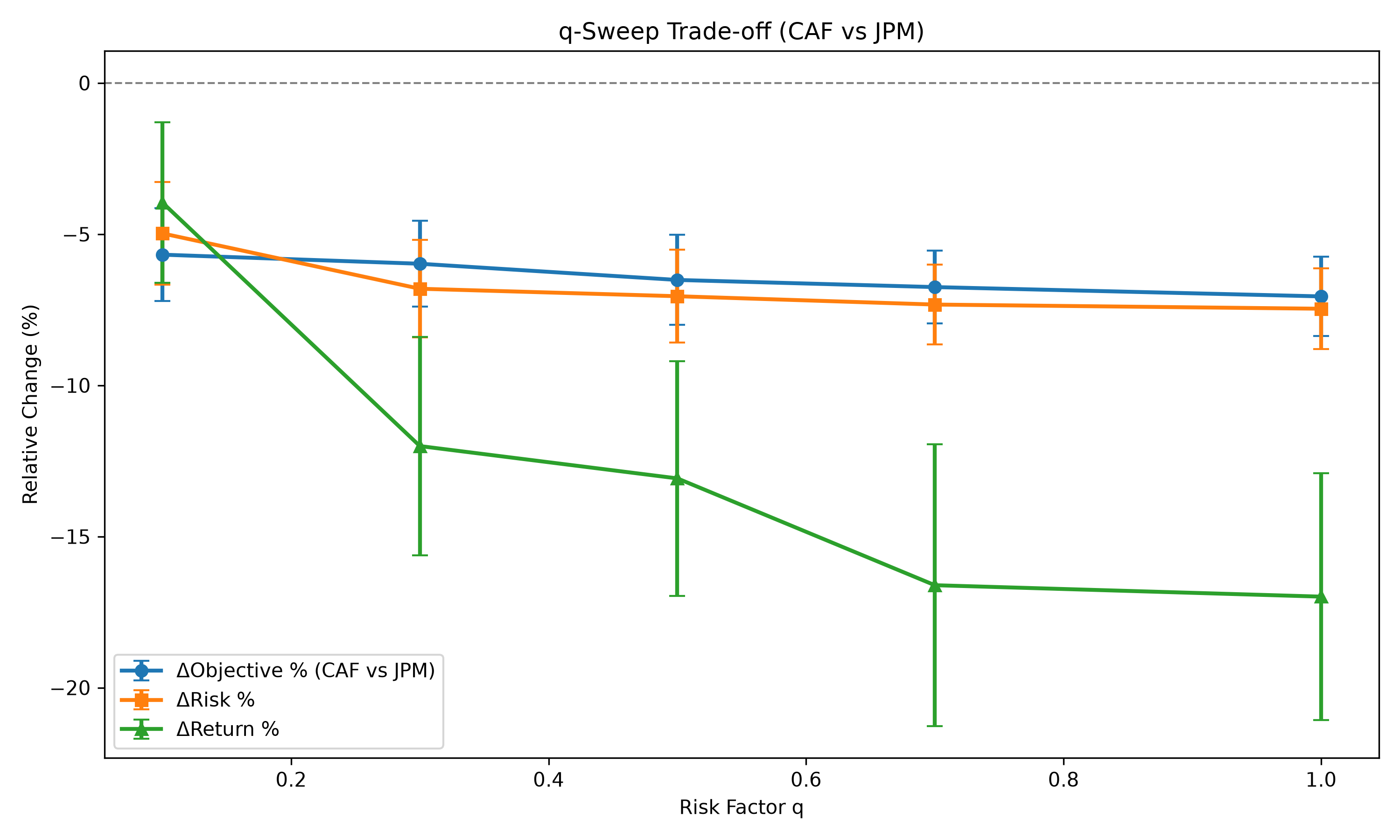}
    \caption{Relative changes versus the static baseline across risk-aversion levels $q$ on the 2016 benchmark under a shared initialization. Error bars denote $\pm 1$ standard deviation.}
    \label{fig:q_sweep}
\end{figure}

Finally, Fig.~\ref{fig:q_sweep} shows that CAF remains beneficial across the tested range of $q$ under a controlled initialization. Although the expected return component may worsen under higher risk aversion, the accompanying risk reduction is large enough to preserve an overall improvement in the scalarized objective $\mathcal{O}(x)$. Table~\ref{tab:benchmark} and Fig.~\ref{fig:q_sweep} are therefore directionally consistent but not numerically identical, because the former reports the dedicated $q=0.5$ benchmark whereas the latter sweeps $q$ during the CAF refinement phase while holding the same initial baseline solution fixed.

At fixed $N=200$, CAF also remains advantageous under different portfolio sparsity regimes. For $K/N=0.25$ ($K=50$), the mean objective improvement reaches $32.95\%$ ($\pm 5.88\%$), while for $K/N=0.50$ ($K=100$) the improvement is $10.52\%$ ($\pm 3.05\%$). This gap is consistent with the view that context-aware correction becomes especially useful in sparse-selection regimes, where local decomposition may be more likely to miss globally coupled risk penalties.

\subsection{Cross-Solver Validation with Gurobi}
We include Gurobi not as an additional deployment baseline, but as a deterministic local backend to assess whether CAF's gains persist beyond stochastic solver effects. This cross-check also anchors the validation in an industrially standard exact-solver regime considered in prior decomposition studies. Accordingly, we repeated the same decomposition and CAF pipeline using Gurobi as the local optimizer. In the present implementation, the local Gurobi solves are deterministic and the upstream spectral clustering path also returned the same partitioning outcome across repeated runs on these instances; the 20 repeated runs therefore serve as a consistency check rather than a stochastic sensitivity study, and all runs returned the same objectives. We report the corresponding means in Table~\ref{tab:gurobi_validation} for consistency with the paper-wide reporting format, but under this deterministic path each mean is identical to the corresponding single-run result. Under the same 2016 benchmark and $K=N/2$, CAF still improved the objective over static decomposition across all tested sizes: $23.45\%$ ($N=100$), $14.29\%$ ($N=200$), $15.15\%$ ($N=300$), and $11.87\%$ ($N=484$).

At fixed $N=200$, the Gurobi-based cardinality sweep shows the same sparsity trend as the SA backend: CAF yields a $43.07\%$ objective improvement at $K/N=0.25$, while the $K/N=0.50$ case coincides with the $N=200$ scaling entry above and therefore is not repeated separately in Table~\ref{tab:gurobi_validation}. Together with Fig.~\ref{fig:time_overhead}, these results support two points: the improvement is not specific to SA stochasticity, and the additional runtime behaves approximately as a constant-factor overhead under early stopping across the tested sizes.

\begin{figure}[tp]
    \centering
    \includegraphics[width=0.98\linewidth]{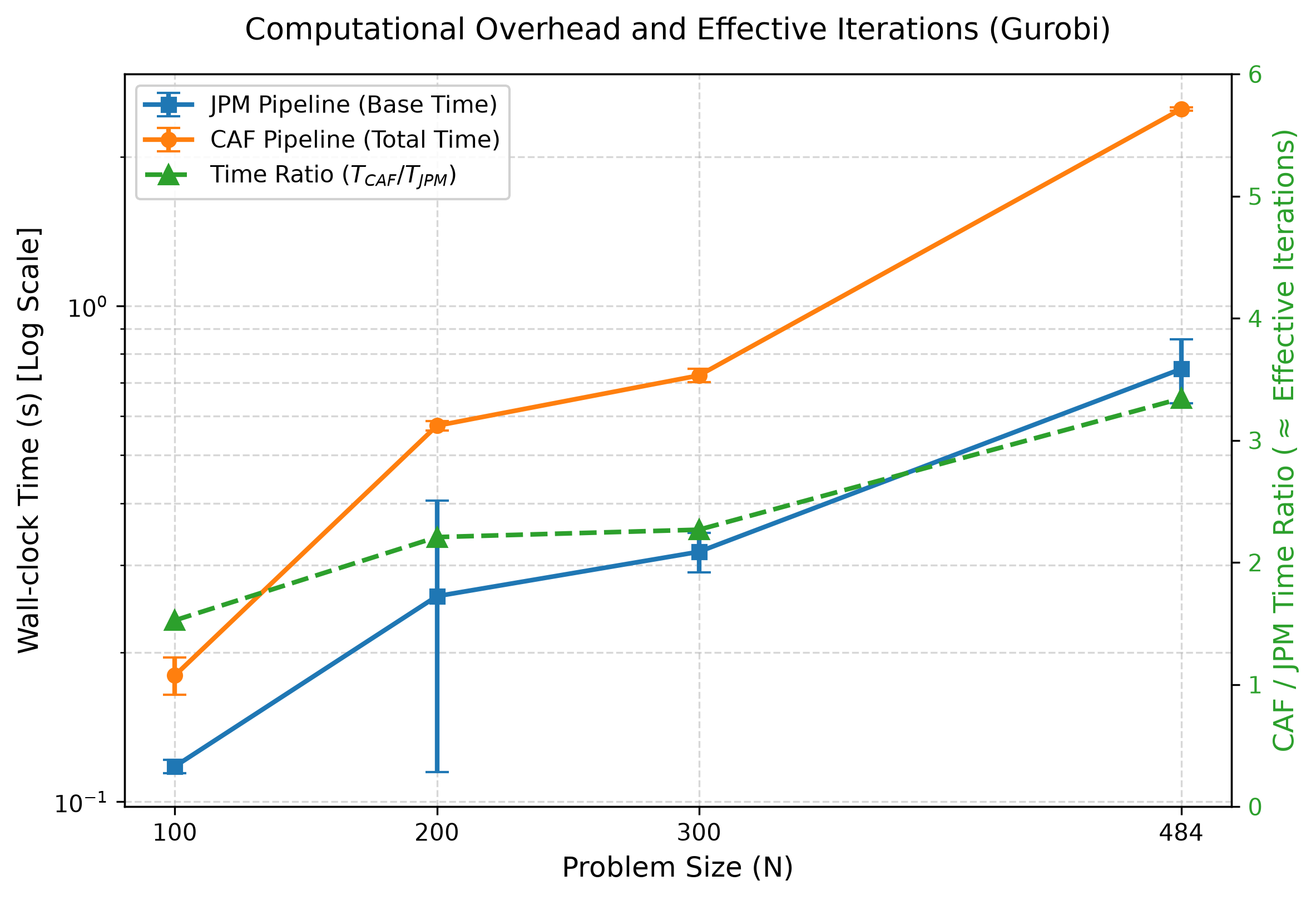}
    \caption{Gurobi runtime overhead and effective iterations across tested sizes. CAF incurs a consistent constant-factor coordination overhead while maintaining positive objective gains throughout the sweep.}
    \label{fig:time_overhead}
\end{figure}

\begin{table}[tp]
\caption{Gurobi cross-solver validation results ($q=0.5$)}
\label{tab:gurobi_validation}
\begin{center}
\resizebox{\columnwidth}{!}{%
\begin{tabular}{lccc}
\toprule
\textbf{Setting} & \textbf{Static Obj. (Mean)} & \textbf{CAF Obj. (Mean)} & \textbf{Mean Improv. (\%)} \\
\midrule
\multicolumn{4}{c}{\textit{Scaling Sweep ($K = N/2$)}} \\
$N=100$ & $0.0310$ & $0.0237$ & $23.45\%$ \\
$N=200$ & $0.2033$ & $0.1742$ & $14.29\%$ \\
$N=300$ & $0.5479$ & $0.4649$ & $15.15\%$ \\
$N=484$ & $1.3138$ & $1.1579$ & $11.87\%$ \\
\midrule
\multicolumn{4}{c}{\textit{Sparsity Sweep ($N = 200$)}} \\
$K=50$ ($0.25N$) & $0.0323$ & $0.0184$ & $43.07\%$ \\
\bottomrule
\end{tabular}%
}
\par\smallskip\footnotesize \textit{Note:} Entries are reported as means over 20 repeated runs for consistency with the paper-wide reporting format; under the current deterministic Gurobi path, all runs returned the same objectives, so these means equal the corresponding single-run results. The $K=100$ ($0.50N$) sparsity case coincides with the $N=200$ scaling entry above (same fixed subset, $K=N/2$) and is omitted to avoid duplication.
\end{center}
\end{table}

\subsection{Context-Induced Landscape Reshaping}

To quantitatively validate the mechanism behind CAF, we conducted a dedicated suppression analysis on the same 2016 benchmark family with $N\in\{100,200,300,484\}$, $K=N/2$, and 10 randomized runs per size (350 community-level samples in total). For each community $C$, we compare the fraction of assets whose local return-side linear term is non-positive before context injection ($\mu_C \le 0$) and after context injection ($\tilde{\mu}_C \le 0$).

CAF increases the suppressed-asset ratio by $87.95$ percentage points on average (std $10.62$), with a 10th--90th percentile interval of $[80.0,\ 100.0]$ percentage points. The effect is similar across scales: $86.69$pp ($N=100$), $88.33$pp ($N=200$), $87.24$pp ($N=300$), and $89.35$pp ($N=484$). Thus, context injection makes a large fraction of local assets less attractive under the adjusted local linear term before local optimization begins, as summarized in Fig.~\ref{fig:sparsity_presolve}.

\begin{figure}[tp]
    \centering
    \includegraphics[width=0.98\linewidth]{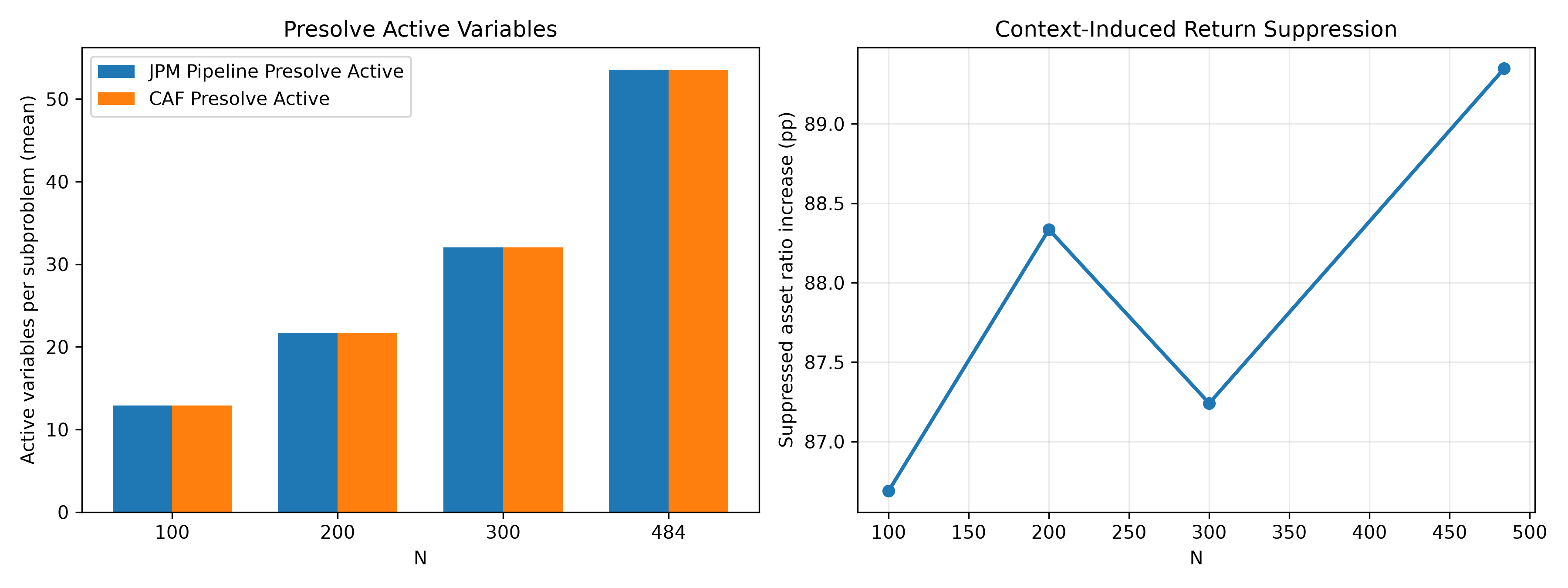}
    \caption{Context-induced landscape reshaping. Left: solver-presolve active-variable counts remain similar. Right: the suppressed-asset ratio $(\tilde{\mu}_C \le 0)$ increases substantially after context injection.}
    \label{fig:sparsity_presolve}
\end{figure}

The left panel provides an auxiliary solver-presolve diagnostic and shows that presolve does not materially reduce the formal variable count in our current formulation. Accordingly, the main observed effect of CAF is not literal variable elimination but a reshaping of the local energy landscape through the adjusted linear term $\tilde{\mu}_C$. For SA or QA backends, this suggests that the practically relevant candidate set becomes more concentrated even when the symbolic model dimension is unchanged, which is consistent with the observed gains being concentrated in the early folding iterations.

We further isolated this mechanism on a frozen $N=40$ CAF state by selecting a hardware-feasible local community with $n=5$ and $K=3$ and comparing the exact local update with and without context injection. On this controlled subproblem, the exact solve without injection gives $\Delta \mathcal{O}=0$ and is therefore rejected, whereas the injected exact solve gives $\Delta \mathcal{O}=-1.0536\times 10^{-4}$ and is accepted. The same injected target is also recovered by warm-start QAOA on the ideal simulator. This frozen-state ablation complements the aggregate suppression analysis above by providing a controlled example in which the injected context term changes the update decision on an individual step, consistent with the broader landscape-reshaping interpretation.

\subsection{Quantum Backend Compatibility on Folded Subproblems}

After establishing the large-scale behavior of CAF with SA and exact classical solvers, we next examine whether the same global-context-injection mechanism remains compatible with quantum simulators and real quantum hardware on folded subproblems. We first consider matched-backend compatibility on folded $N=40$ instances and then turn to frozen local hardware validation on selected warm-started subproblems. Here QAOA serves only as a small-scale gate-based compatibility study rather than as a primary benchmark backend. To define a reproducible evaluation target, we systematically scanned global portfolio sizes $N \in \{24, 32, 40, 48, 56, 64, 72, 80\}$ across 5 independent random seeds. We selected $N=40$ as the representative test point because it is the largest global dimension in this scan for which the maximum detected community size remained consistently bounded by $n \le 16$ across all 5 evaluated seeds. This bound is required by the exact statevector simulator used as the local subroutine, whose memory cost grows as $O(2^n)$.

To assess how the gate-based compatibility evidence changes with circuit depth, we conducted matched-backend comparisons at depths $p=2$ and $p=3$ across 20 stochastic runs, where the static baseline and CAF use the same local QAOA setting within each depth. Here, ``matched'' refers to using the same local solver family and hyperparameter budget on the same folded instance; it does not denote a direct runtime comparison against the SA main benchmark. Plain QAOA (no warm-start) improves the global objective in all 20 runs at both depths. At $p=2$, the mean absolute objective delta is $-5.95 \times 10^{-4} \pm 2.72 \times 10^{-4}$, with a mean risk reduction of $13.56\% \pm 4.39\%$ and a mean return change of $-7.83\% \pm 2.66\%$. At $p=3$, the mean absolute objective delta is $-4.91 \times 10^{-4} \pm 1.78 \times 10^{-4}$, with a mean risk reduction of $13.84\% \pm 3.74\%$ and a mean return change of $-8.71\% \pm 2.70\%$. Increasing the depth from $p=2$ to $p=3$ does not yield a statistically meaningful additional gain on these subproblems, indicating that $p=2$ already suffices here.

We next evaluated an annealing-style simulated backend using D-Wave's \texttt{PathIntegralAnnealingSampler} \cite{martonak2002pimc,dwave_samplers_docs}, a path-integral Monte Carlo (PIMC) realization of simulated quantum annealing, under the same matched CAF-vs-static comparison. CAF improves 18 of 20 runs, with a mean objective delta of $-6.10 \times 10^{-4} \pm 4.36 \times 10^{-4}$, a $10.68\% \pm 5.62\%$ risk reduction, and a $-5.68\% \pm 3.85\%$ return change. The QAOA and SQA results in Table~\ref{tab:qaoa_feasibility} are thus compatibility evidence, not a performance benchmark against the SA mainline: context injection remains operable across both gate-based and annealing-style local solvers while preserving the outer-loop acceptance logic. Here the gate-based columns use plain QAOA (no warm-start); the hardware-in-the-loop study below uses a separate warm-started protocol. Because instance scale and backend dynamics differ from the $N=484$ benchmark, the return percentages are only supplementary small-instance indicators; the stable cross-setting quantity is the absolute objective delta, on the order of $10^{-4}$.

\begin{table}[tp]
\caption{Quantum Backend Compatibility on a Folded $N=40$ Instance}
\label{tab:qaoa_feasibility}
\begin{center}
\resizebox{\columnwidth}{!}{%
\begin{tabular}{@{}lccc@{}}
\toprule
\textbf{Metric} & \textbf{QAOA ($p=2$)} & \textbf{QAOA ($p=3$)} & \textbf{SQA (PIMC)} \\
\midrule
Global Assets ($N$) & 40 & 40 & 40 \\
Max Subproblem Size ($n$) & 16 & 16 & 16 \\
Improving Runs (Wins) & 20 / 20 & 20 / 20 & 18 / 20 \\
Mean Objective $\Delta$ & $-5.95 \times 10^{-4} \pm 2.72 \times 10^{-4}$ & $-4.91 \times 10^{-4} \pm 1.78 \times 10^{-4}$ & $-6.10 \times 10^{-4} \pm 4.36 \times 10^{-4}$ \\
Mean Risk Reduction & $13.56\% \pm 4.39\%$ & $13.84\% \pm 3.74\%$ & $10.68\% \pm 5.62\%$ \\
Mean Return Change & $-7.83\% \pm 2.66\%$ & $-8.71\% \pm 2.70\%$ & $-5.68\% \pm 3.85\%$ \\
\bottomrule
\end{tabular}%
}
\par\smallskip\footnotesize \textit{Note:} Each column compares CAF with a matched static baseline under the same local backend and hyperparameter budget on the folded $N=40$ instance. Positive ``Mean Risk Reduction'' means lower risk under CAF. ``Mean Return Change'' is reported only as a supplementary small-instance indicator and is not intended for direct comparison with the main $N=484$ SA benchmark.
\end{center}
\end{table}

We next consider frozen local hardware validation on Origin Wukong 180 (backend \texttt{WK\_C180}). These warm-started hardware-in-the-loop results are reported as a local-update executability demonstration rather than as a same-protocol comparison with Table~\ref{tab:qaoa_feasibility}. The gate-based workflow is warm-started throughout; for each frozen hardware target, we first optimize the QAOA parameters once on the ideal simulator and then execute the same ansatz with the same $\theta^\star$ on hardware. Although the platform supports amend-based readout mitigation \cite{nation2021m3}, the reported candidate bitstring and $\Delta \mathcal{O}$ are computed from raw measurement counts, so the executability claim reflects what the hardware actually sampled.

The decoder selects, directly from the raw measurement counts, the feasible bitstring with the best local surrogate objective---the same rule used in the warm-started simulator protocol. In the ideal simulator the warm-start ansatz concentrates essentially all amplitude on the optimum; on hardware, noise disperses this concentration. For the $n=5$ frozen subproblem the optimum survives in only a small fraction of raw shots but is still selected as the best feasible candidate, so the hardware candidate matches the simulator optimum and yields the same accepted global objective decrease. Larger probes ($n=7,9,10$) no longer sample the optimum reliably under raw shot counts, so we restrict the stable hardware claim to $n=5$.

\begin{table}[tp]
\caption{Additional Frozen Hardware Validation on Origin Wukong 180 (backend \texttt{WK\_C180})}
\label{tab:hardware_local_updates}
\begin{center}
\resizebox{\columnwidth}{!}{%
\begin{tabular}{@{}lcccc@{}}
\toprule
\textbf{Frozen Subproblem} & $\mathbf{\Delta O_{\mathrm{sim}}}$ & $\mathbf{\Delta O_{\mathrm{hw}}}$ & \textbf{Accepted} & \textbf{Hamming} \\
\midrule
$n=5$, $K=3$ & $-1.0536 \times 10^{-4}$ & $-1.0536 \times 10^{-4}$ & Yes & 0 \\
\bottomrule
\end{tabular}%
}
\par\smallskip\footnotesize \textit{Note:} The reported $n=5$ update uses the warm-started local-update protocol on a frozen CAF state, executed on backend \texttt{WK\_C180}, and decodes the candidate from raw measurement counts under the same surrogate-objective rule as the simulator. Larger probes ($n=7,9,10$) do not reliably sample the optimum under raw shot counts and are not claimed as stable hardware evidence.
\end{center}
\end{table}

Thus, the hardware evidence supports a limited but concrete claim: under raw shot counts, a warm-started local CAF update can be executed on real hardware for the small $n=5$ frozen subproblem, while larger probes ($n=7,9,10$) do not reliably preserve the optimum under current hardware noise, so the stable executability claim is restricted to $n=5$.

\section{Conclusion}

The Context-Aware Folding (CAF) architecture shows that part of the global information lost after static partitioning can be reintroduced through global context injection. On the 484-asset 2016 benchmark ($q=0.5$), CAF improves the scalarized objective by $6.59\%$ in all 20 paired runs, and by $0.2579\%$ (17/20 wins) on an additional 2018 panel, so the gain persists across market regimes though its magnitude varies with the headroom left by the static decomposition. Cross-solver Gurobi experiments show the same advantage across scale and sparsity, confirming the gain stems from coordination rather than SA stochasticity, and the suppression analysis indicates CAF reshapes the local landscape rather than reducing the formal problem dimension. Matched $N=40$ studies show CAF remains compatible with quantum simulators and real quantum hardware, and the frozen hardware study demonstrates a warm-started local update on Wukong 180 at $n=5$ (larger probes do not reliably sample the optimum under raw shot counts). More broadly, CAF indicates that the gap between one-shot static decomposition and full-matrix optimization admits a lightweight coordination layer that is largely solver-agnostic, with the quantum contribution kept at the level of local compatibility and limited executability rather than advantage. CAF remains dependent on the upstream community partition, and this work does not benchmark concurrent execution on real hardware or address out-of-sample financial evaluation; future work will extend local-solver comparisons, hardware targets, and asynchronous context updates.

\section*{Acknowledgment}
The authors used AI-based language assistance tools during manuscript preparation for wording refinement and editorial suggestions. All technical content and final manuscript decisions were verified and approved by the authors.

\FloatBarrier

\bibliographystyle{IEEEtran}
\bibliography{references}

\end{document}